\documentclass[singlecolumn,preprintnumbers,amsmath,notitlepage,nofootinbib,longbibliography,superscriptaddress]{revtex4-1}

\usepackage[utf8]{inputenc}
\usepackage[english]{babel}
\usepackage{comment}

\usepackage[T1]{fontenc}
\usepackage{amssymb}
\usepackage{amsthm}
\usepackage{verbatim} 

\usepackage{cleveref} 
\usepackage[ruled,lined]{algorithm2e}

\usepackage{natbib}
\newtheorem{lemma}{Lemma}
\newtheorem{definition}{Definition}
\newtheorem{theorem}{Theorem}
\newtheorem{corollary}{Corollary}
\newtheorem{observation}{Observation}

\newcommand{\Prob}{\operatorname{Pr}}

\usepackage{tikz}
\newcommand*\circled[1]{\tikz[baseline=(char.base)]{
            \node[shape=circle,draw,inner sep=2pt] (char) {#1};}}

\usepackage{bbm}
\newcommand{\indexf}{\mathbbm{1}}

\newcommand{\bsn}{\{0,1\}^n}
\newcommand{\Zq}{\mathbb{Z}_q}

\newcommand{\EVAL}{\operatorname{EVAL}}
\newcommand{\MQ}{\operatorname{MQ}}
\newcommand{\REX}{\operatorname{REX}}
\newcommand{\KPRFC}{\left\{F_{\theta} : \mathcal{K}_{\theta} \times \mathcal{X}_{\theta} \rightarrow \mathcal{Y}_{\theta} | \theta \in \Theta\right\}}

\newcommand{\CN}{\mathcal{D}}
\newcommand{\DPK}{D_{(\theta, k)}}
\newcommand{\Fpgga}{F_{(p,g,g^a)}}

\newcommand{\QRp}{\operatorname{QR}_p}

\newcommand{\AD}{\mathcal{A}_{\mathcal{D}}}
\newcommand{\AF}{\mathcal{A}_{\mathcal{F}}}

\newcommand{\ryan}[1]{\textcolor{teal}{#1}}
\newcommand{\niklas}[1]{\textcolor{olive}{[#1]}}
\newcommand{\je}[1]{\textcolor{blue}{[#1]}}

\begin{document}

\title{A super-polynomial quantum-classical separation for density modelling}

\author{Niklas Pirnay}
\address{Electrical Engineering and Computer Science, Technische Universit{\"a}t Berlin, Berlin, 10587, Germany}
\author{Ryan Sweke}
\altaffiliation{Currently at IBM Quantum, Almaden Research Center, San Jose, CA 95120, USA.}
\address{Dahlem Center for Complex Quantum Systems, Freie Universit{\"a}t Berlin, Berlin, 14195, Germany}
\author{Jens Eisert}
\address{Dahlem Center for Complex Quantum Systems, Freie Universit{\"a}t Berlin, Berlin, 14195, Germany}
\address{Fraunhofer Heinrich Hertz Institute, 10587 Berlin, Germany}
\author{Jean-Pierre Seifert}
\address{Electrical Engineering and Computer Science, Technische Universit{\"a}t Berlin, Berlin, 10587, Germany}
\address{Fraunhofer SIT, D-64295 Darmstadt, Germany}

\date{\today}

\begin{abstract}
 Density modelling is the task of learning an unknown probability density function from samples, and is one of the central problems of unsupervised machine learning. In this work, we show that there exists a density modelling problem for which fault-tolerant quantum computers can offer a super-polynomial advantage over classical learning algorithms, given standard cryptographic assumptions. Along the way, we provide a variety of additional results and insights, of potential interest for proving future distribution learning separations between quantum and classical learning algorithms. Specifically, we (a) provide an overview of the relationships between hardness results in supervised learning and distribution learning, and (b) show that any \textit{weak} pseudo-random function can be used to construct a classically hard density modelling problem. The latter result opens up the possibility of proving quantum-classical separations for density modelling based on weaker assumptions than those necessary for pseudo-random functions.
\end{abstract}

\maketitle

\section{Introduction}

The task of learning a representation of a probability distribution from samples is of importance in a wide variety of contexts, from the natural sciences to industry. As such, a central focus of modern machine learning is to develop algorithms and models for this task. Of particular importance is the distinction between \textit{density modelling} and \textit{generative modelling}. In density modelling the task is to learn an \textit{evaluator} for a distribution -- i.e., a function which on input of a sample returns the probability weight assigned to that sample by the underlying distribution. As such, density modelling is sometimes referred to, as we do here, as \textit{evaluator learning}. In generative modelling, the task is to learn a \textit{generator} for a distribution -- i.e., a function which given a (uniformly) random input seed outputs a sample with probabilities according to the target distribution. As a result, generative modelling is sometimes referred to as \textit{generator learning}. As an example, in a generative modelling problem the goal might be to generate images of cats or dogs, while the associated density modelling problem would be to evaluate the probability that an image depicts a cat or a dog.
It is important to stress that these two learning tasks are indeed fundamentally different. That is to say, efficiently learning a generator for a distribution does not imply efficiently learning an evaluator and vice versa~\cite{kearns_learnability_1994}. 

Given the incredible success of modern machine learning, and the rapidly increasing availability of \textit{quantum} computational devices, a natural question is whether or not quantum devices can provide any advantage in this domain~\cite{biamonte_qml_2017,dewolf_qlearntheo_2017,lloyd_qml_2013,carleo_mlphys_2019}. Most current research in this direction is of a heuristic nature \citep{benedetti_pqc_2019, cerezo2021variational}. However, there also exists an emerging body of results which provide examples of machine learning tasks for which one can prove rigorously a meaningful separation between the power of classical and quantum computational devices \citep{servedio_equivalences_2004,  dunjko_exponential_2018, liu_rigorous_2021, sweke_generator_2021}, even though results on such rigorous separations are still rather
scarce~\cite{onseparations}. 
One of these, the work of 
Ref.~\citep{sweke_generator_2021} has shown rigorously that one can obtain a quantum advantage in \textit{generative modelling} by constructing a distribution class which is (a) provably hard to generator learn classically, but (b) efficiently generator learnable using a fault-tolerant quantum computer. Importantly however, Ref.~\citep{sweke_generator_2021} has not addressed the related task of \textit{density modelling}. 
  
In this work, we close this gap, by showing that the class of probability distributions constructed in 
Ref.~\citep{sweke_generator_2021} in fact also allows one to demonstrate a super-polynomial quantum-classical separation for density modelling. Additionally, along the way we provide a variety of insights and additional results, which may be of independent interest for constructing future quantum-classical separations in distribution learning. More specifically, in this work we do the following:
  \begin{enumerate}
      \item Any quantum-classical separation requires a proof of classical hardness. The generative modelling separation of 
      Ref.~\citep{sweke_generator_2021} relies crucially on a result from 
      Ref.~\citep{kearns_learnability_1994} which shows that from any \textit{pseudo-random function} (PRF) one can construct a distribution class which is provably hard to generator learn classically. We strengthen this fundamental tool, by showing that for the case of density modelling, any \textit{weak PRF} can be used to construct a distribution class which is provably hard to evaluator learn classically. This opens up the door for proving classical hardness results for density modelling, based on weaker assumptions than those necessary for candidate PRF constructions. In particular, the hope is that one may be able to prove classical hardness using assumptions which do not immediately also rule out the possibility of efficient learning algorithms running on \textit{near-term quantum devices}.
      \item We prove a super-polynomial quantum-classical separation for density modelling using the distribution class from 
      Ref.~\citep{sweke_generator_2021}. As the distribution class from  
      Ref.~\citep{sweke_generator_2021} is constructed from a PRF, the classical hardness follows immediately given the above mentioned insight that even weak PRFs are sufficient for density modelling hardness. As such, what remains is to provide an efficient quantum evaluator-learner for this distribution class, and we show that a simple modification of the quantum generator-learner from Ref.~\citep{sweke_generator_2021} is sufficient to achieve this. 
      \item The majority of work in computational learning theory has been focused on the task of \textit{supervised learning} Boolean functions. As such, it is natural to ask to which extent hardness results and quantum-classical separations in supervised learning can be leveraged to obtain separations for distribution learning. We provide an overview of the extent to which this is or is not possible, for both generative and density modelling. Once again, the hope is that this provides a toolbox for proving future quantum-classical separations in distribution learning.
  \end{enumerate}
  
This work is structured as follows: We start below in Section~\ref{sec:background} by providing some essential definitions and background from both computational learning theory and cryptography. We note that as this work to a large extent generalizes and extends 
Ref.~\citep{sweke_generator_2021}, we do not provide all necessary background here, and we refer often to 
Ref.~\citep{sweke_generator_2021} for a variety of definitions and constructions.  Given the necessary background we proceed in Section~\ref{sec:results} to present a variety of techniques -- both known and novel -- for proving hardness results in distribution learning from hardness results in supervised learning. Of particular interest is Section~\ref{subsec:classhard}, in which we show that one can use any \textit{weak} PRF to construct a distribution class which is classically hard to evaluator learn.
Using these tools, we then show in Section~\ref{subsec:quanteasy} a super-polynomial quantum-classical separation for density modelling, using the distribution class from 
Ref.~\citep{sweke_generator_2021}. Finally, we conclude in Section~\ref{sec:discussion} with a discussion and outlook.

\section{Background}
\label{sec:background}

\noindent To show a quantum-classical learning separation, on the highest level, one needs to
prove two things: one has to prove classical learning hardness and show efficiency of quantum learning.
To introduce the necessary formalism, we will start in this section by providing an overview of the PAC framework for learning both functions and distributions.
Given this, we will then present a construction from Ref.~\citep{kearns_learnability_1994} which allows one to define distribution classes from function classes in a way which facilitates the conversion of function learning hardness to distribution learning hardness.
Finally, we introduce weak-secure
pseudo-random functions, which will later be used to construct distribution classes, via the aforementioned construction from Ref.~\cite{kearns_learnability_1994}, for which the density modelling problem is provably hard for classical learning algorithms. In what follows, we denote:
\begin{itemize}
  \item $\indexf(x,y)$ : the index function evaluating to $1$ if and only if $x=y$,
  \item $x \sim U(\mathcal{X})$ : sample $x$ from the uniform distribution over a set
  $\mathcal{X}$ (sometimes $\mathcal{X}$ is omitted if $\mathcal{X}$ is clear from
  the context),
  \item $\mathrm{poly}(a,b)$ : any polynomial in $a$ and $b$ or sometimes also meaning the set of all polynomials in $a,b$, 
  \item $\bsn$ : the set of bit strings of length $n$,
  \item $a\|b$ :  the concatenation of bit strings $a,b$,
  \item $1^n$  : the bit string consisting of $n$ $1$'s,
  \item $D(x)$ : the probability mass assigned to $x$, if $D$ is a probability distribution,
   \item $d_{TV}(P,Q)$ : the total variation distance between distributions $P$ and $Q$,
  \item $\Zq$ : the residue class ring $\mathbb{Z} / q\mathbb{Z}$.
\end{itemize}

Before introducing the formalism for analyzing distribution learning, we introduce Valiant's PAC learning framework for function learning~\citep{valiant_pac_1984}, which since its proposal is the standard framework for  rigorously analyzing the complexity of supervised learning problems~\cite{kearns1994introduction}.
In it we are concerned with learning some class of functions that map $n$ bits to $m$ bits. At a high level, for any such target function in the class, when given some sort of oracle access to the unknown target function, a learning algorithm should with high probability, output a hypothesis function that is close to the target function.

For this function learning task, we distinguish between two different types of oracle access to the function $f$ that is to be learned\footnote{We note that one can consider many other types of oracle access as well, such as for example, statistical query access~\cite{sqmodel}.}. Firstly, the membership query oracle to $f$, $\MQ(f)$, which when queried with $x$ yields the tuple $(x,f(x))$. We denote this via
\begin{equation}
\mathrm{query}[\MQ(f)](x) = (x,f(x)).
\end{equation}
The membership query access corresponds to the ability to evaluate $f$ on chosen points. We sometimes refer to the membership query oracle in general without any fixed function simply as $\MQ$.
Secondly, the $\eta$-noisy random example oracle $\REX(f,P,\eta)$ to $f$ is defined via
\begin{equation}
    \mathrm{query}[\REX(f,P,\eta)] = \begin{cases}
      (x, f(x)) \text{ with probability } P(x)(1-\eta) \\
      (x, \lnot f(x)) \text{ with probability } P(x)(\eta)
    \end{cases} \text{,}
\end{equation}
where $P$ is a probability distribution over inputs, $\eta \in [0,1]$ is a noise rate and $\lnot f(x)$ is any element in the image space of $f$ except $f(x)$. If $\eta=0$, we also write $\REX(f,P)$. At a high level, this oracle generates random (possibly noisy) input output tuples from $f$.  If we refer to the random example oracle in general, without any fixed function, we write $\REX(P,\eta)$. Algorithmically, we consider a query to $\MQ$ or $\REX$ to take unit time.

We can now give the definition of a PAC learning algorithm for function classes. Note that here we use the notation $O(f,D)$ to denote some oracle, which might be either $\MQ(f)$ or $\REX(f, D)$.

\begin{definition}[$(\epsilon, \delta, O)$-PAC function learner for $\mathcal{F}$] Let $\mathcal{F}$ be a class of functions, with $f:\{0,1\}^n\rightarrow \{0,1\}^m$ for all $f\in \mathcal{F}$ . Given some fixed $\epsilon, \delta\in (0,1)$, an algorithm $\mathcal{A}$ is 
  an $(\epsilon, \delta, O, D)$-PAC function
  learner $\mathcal{F}$, if for all
  $f \in \mathcal{F}$, when given oracle
  access $O(f,D)$, with
  probability at least $1-\delta$, $\mathcal{A}$
  outputs a hypothesis $h$ satisfying
  \begin{align}
    \Prob_{x \sim D} \left[ f(x) \neq h(x) \right] \leq \epsilon \text{.}
  \end{align}
  The algorithm $\mathcal{A}$ 
  is an 
  $(\epsilon, \delta, O)$-PAC
  function learner for $\mathcal{F}$ if it is a
  $(\epsilon, \delta, O, D)$-PAC function learner for
  $\mathcal{F}$ for all distributions $D$.
  We call $\mathcal{A}$ an efficient $(\epsilon, \delta, O)$-PAC
  function learner for $\mathcal{F}$ if the time complexity of $\mathcal{A}$ is $O(poly(n))$.
  We call $\mathcal{F}$ $(\epsilon, \delta, O)$-PAC-hard, if there exists no efficient $(\epsilon, \delta, O)$-PAC function learner for it.
\end{definition}
The above definition refers to fixed accuracy and probability parameters $(\epsilon,\delta)$, but we note that if these parameters are considered as variables, then an efficient learner is taken as one with time complexity $O(\mathrm{poly}(n,1/\epsilon,1/\delta))$. In this case, the algorithm will be efficient with respect to the definition above for any $\epsilon,\delta = \Omega(1/\mathrm{poly}(n))$. Additionally, we stress that the learning algorithm could be either classical or quantum. Indeed, as we will see in this work, it is possible that there exists an efficient quantum learning algorithm for a given class, but no efficient classical learning algorithm.

%
We would now like to 
generalize the PAC framework for learning functions to the natural and important problem of learning distributions.
To formulate this problem rigorously, it is necessary to first introduce the different possible representations of a distribution that one might want to learn, namely \textit{generators} and \textit{evaluators}:
%
\begin{definition}[Generator and evaluator for $D$]
  Let $D$ be a discrete probability distribution over $\bsn$.
  A generator for $D$ is any function $\mathrm{GEN}_{D}:\{0,1\}^m \rightarrow \bsn$ that
  on uniformly random inputs outputs samples according to $D$, i.e.,
  \begin{align}
    \Prob_{x \sim U(\{0,1\}^m)} \left[ \mathrm{GEN}_{D}(x) = y \right] = D(y) \text{.}
  \end{align}
  An evaluator for $D$ is any function $\EVAL_D:\bsn \rightarrow [0,1]$ that evaluates the
  probability mass assigned to a event with respect to $D$, i.e.,
  \begin{align}
    \EVAL_D(x) = D(x) \text{.}
  \end{align}
\end{definition}
We note that evaluating and generating are indeed two distinct tasks and in general, the ability to do the one does not imply the ability to do the other.
To avoid any complexity-theoretic loopholes, during the course of this work, we assume that any evaluators or generators of interest are computable in time $poly(n)$.
%
With the definition of a generator and an evaluator of a distribution at hand, we can now define PAC learners for distributions.
\begin{definition}[$(\epsilon, \delta)$-PAC generator and evaluator learner for $\mathcal{D}$]
  Let $\mathcal{D}$ be a class of discrete
  probability distributions over $\bsn$. Given some  fixed $\epsilon,\delta\in (0,1)$
  an algorithm $\mathcal{A}$
  is an $(\epsilon, \delta)$-PAC (a) generator (GEN) or (b) evaluator (EVAL) learner of $\mathcal{D}$, if for all
  $D \in \mathcal{D}$, when given access to samples from $D$, with probability at least
  $1-\delta$, $\mathcal{A}$ outputs a
  (a) generator or (b) evaluator for some distribution $D'$, satisfying
  \begin{align}
    d_{TV}(D,D')\leq \epsilon.    
  \end{align}
  We call $\mathcal{A}$ an efficient $(\epsilon,\delta)$-PAC (generator or evaluator) learner for $\mathcal{D}$ if its time complexity is $O(\mathrm{poly}(n))$. We call $\mathcal{D}$ $(\epsilon,\delta)$-PAC (generator or evaluator) hard if there is no efficient $(\epsilon,\delta)$-PAC (generator or evaluator) learner for $\mathcal{D}$.
\end{definition}

%
In this work, we aim at proving a quantum-classical separation for distribution learning and we want to do this by leveraging known classical hardness results for learning functions.
In order to carry PAC function learning hardness results to the distribution learning regime, we 
use a generalization of a construction from 
Ref.~\citep{kearns_learnability_1994}, which allows one to define from any function a corresponding distribution. More specifically, we consider \textit{induced distributions} defined as follows:
\begin{definition}[Induced distribution of $f$] \label{def:induced_dist}
  For any function $f:\bsn \rightarrow \{0,1\}^m$, we define the induced distribution $D_{f, P, \eta}$ as the
  discrete probability distribution over $\{0,1\}^{n+m}$ via
  \begin{align}
    D_{f,P,\eta}(x\|y) = \begin{cases}
    P(x)(1-\eta) \text{, if } f(x)=y \\
    P(x)(\eta/(2^m -1)) \text{, else}
    \end{cases}
  \end{align}
  for $x \in \bsn$ and $y \in \{0,1\}^m$ and $\eta \in [0,1]$ and $P$ any probability distribution over $\bsn$. If $\eta=0$, we simply write $D_{f, P}$.
  Similarly, we define the induced distribution class of the function class $\mathcal{F}$ by $\mathcal{D}_{\mathcal{F},P,\eta} = \{D_{f,P,\eta} | f \in \mathcal{F}\}$ and write $\mathcal{D}_{\mathcal{F},P}$ if $\eta=0$.
\end{definition}
We note that the induced distributions defined above are constructed precisely to allow a direct correspondence between oracle access to the function and sample access to the induced distribution. In particular, we note that a query to $\REX(f,P,\eta)$ is precisely the same as drawing a sample from $D_{f,P,\eta}$.


The final background ingredient we require is that of \textit{pseudo-random functions}, which as we will soon see, allows us to prove distribution learning hardness results for the associated induced distributions.
Intuitively, a pseudo-random function $f$ is one that cannot be distinguished from a completely
random function, by any polynomial-time algorithm that has oracle access to $f$, with non-negligible probability.
Before giving the definition of pseudo-random functions, let $\Theta$ be a parameter set, for which there exists an instance generation algorithm $\mathcal{I G}$ which on input $1^n$ outputs some ``size $n$'' parameter $\theta \in \Theta$.
We call $\mathcal{I G}$ efficient and $\Theta$ efficiently sampleable if $\mathcal{I G}$ runs in time $O(\mathrm{poly}(n))$.
For more details on why we require this efficiently sampleable parameter set, we refer the reader to 
Ref.~\citep{sweke_generator_2021}.
\begin{definition} [Pseudo-random function collection] \label{def:prfs}
A set of efficiently computable functions
\begin{align}
\KPRFC
\end{align}
is called a (a) classic-secure or (b) weak-secure pseudo-random function collection if for all classical probabilistic polynomial time algorithms $\mathcal{A}$, all polynomials $p$, and all sufficiently large $n$, it holds that
 \begin{align}
\left|\Prob_{\underset{\theta \leftarrow \mathcal{I G}\left(1^{n}\right)} {k \sim U\left(\mathcal{K}_{\theta}\right)}} \left[\mathcal{A}^{O\left(F_{\theta}(k, \cdot)\right)}(\theta)=1\right]
-
\Prob_{\underset{\theta \leftarrow \mathcal{I G}\left(1^{n}\right)} {R \sim U\left(F: \mathcal{X}_{\theta} \rightarrow \mathcal{Y}_{\theta}\right)}} \left[\mathcal{A}^{O(R)}(\theta)=1\right]\right|<\frac{1}{p(n)}
 \end{align}
where $U\left(F: \mathcal{X}_{\theta} \rightarrow \mathcal{Y}_{\theta}\right)$ denotes the uniform
distribution over all functions from $\mathcal{X}_{\theta}$ to $\mathcal{Y}_{\theta}$, $\mathcal{K}_\theta$
denotes the key space, $\mathcal{I G}$ is the efficient instance generation algorithm for the parameters $\theta$ and $\mathcal{A}$ is given oracle access to (a) $O(f)=\MQ(f)$ or (b)
$O(f)=\REX(f, U)$.
\end{definition}
Note the core statement of the definition above: Any polynomial time algorithm $\mathcal{A}$
with access to the oracle $O(f)$ cannot determine with non-negligible probability whether $f$ was drawn from the function collection, or is a truly
random function.
While a \textit{classic-secure} PRF cannot be distinguished from a
random function using membership query access to the function, a \textit{weak-secure} PRF cannot be distinguished from a
random function using random example access.
Since there exists an algorithm that can simulate random example queries
using membership queries, membership query access is more powerful
than random example access and any classic-secure PRF is also weak-secure.


\section{From supervised learning to distribution learning}
\label{sec:results}

\noindent As we have mentioned, showing a quantum-classical distribution learning separation requires us to show two things: classical hardness and efficiency of quantum learning (for the same distribution learning task).
For showing the former, a variety of techniques have been used previously, most of which exploit either PAC-hard functions or PRFs, primarily through the "function to distribution construction" in Definition \ref{def:induced_dist} of the previous section.
In order to consolidate and make explicit these techniques, we provide in this section an overview of known results and methods, as well as two extensions and generalizations. 
In particular, we first provide a theorem which abstracts and generalizes the technique of translating PAC function learning hardness to PAC evaluator learning hardness (used implicitly in Ref.~\cite{kearns_learnability_1994}) for functions that map to $m=O(\log(n))$ bits, even in the case of noisy random examples. Additionally, we then provide a theorem which shows that one can prove evaluator learning hardness for distributions induced by \textit{weak} PRFs. This strengthens, and makes applicable to density modelling, the technique used in Ref.~\cite{kearns_learnability_1994} to prove hardness of \textit{generative} modelling from PRFs.
Table~\ref{table:overview} below puts our unique contributions in the context of prior work. We note that the primary focus of our work is on \textit{density} modelling (i.e., evaluator learning) and we refer to Ref.~\cite{xiao_learning_2010} for a similar study focused on generative modelling, which considers additional "function to distribution" constructions than the one presented here.
\begin{table}[h]
\centering
\begin{tabular}{l|l|l}
Function class $\mathcal{F}$      & EVAL learning $\mathcal{D}_{\mathcal{F},P}$ & GEN learning  $\mathcal{D}_{\mathcal{F
},P}$ \\ \hline
weak PRFs                  & Hard (Theorem \ref{theo:classical_eval_hardness})     & Open question               \\ \hline
PRFs                       & Hard (Corollary of Theorem \ref{theo:classical_eval_hardness})     & Hard 
(Ref.~\citep{kearns_learnability_1994} and Ref.~\cite{xiao_learning_2010}) \\ \hline
PAC-hard $m=1$  &  Hard (follows from Corollary \ref{cor:warmup}) & Not necessarily hard (Ref.~\citep{xiao_learning_2010}) \\ \hline
PAC-hard $m=O(\log(n))$          & Hard (implicit in 
Ref.~\citep{kearns_learnability_1994} -- explicit in Corollary \ref{cor:warmup})   & Open question               \\ \hline
PAC-hard $m=\Omega(\log(n))$          & Open question  & Open question 
\end{tabular}
\caption{Given some hard-to-learn function class $\mathcal{F}$ that contains functions mapping from $n$ bits to $m$ bits, can we obtain EVAL or GEN learning hardness for $\mathcal{D}_{\mathcal{F},P}$?
The table above places our contributions within the context of prior work on this question.
Firstly, in Theorem~\ref{theo:warmup} we generalize the implicit techniques from Ref.~\citep{kearns_learnability_1994} to show in Corollary~\ref{cor:warmup} that PAC-hard functions translate to EVAL learning hardness for functions with $m=O(\log(n))$ output bits.
Secondly, in Theorem~\ref{theo:classical_eval_hardness} we show that EVAL learning hardness can generally be obtained from weak PRFs.}
\label{table:overview}
\end{table}

\subsection{Learning distributions induced by $m=O(\log(n))$-functions}
\label{subsec:warmup}

We begin by establishing (in Theorem~\ref{theo:warmup}) a direct relationship between
PAC learning a function class and PAC evaluator learning the induced distribution class
for functions $f:\{0,1\}^n\rightarrow\{0,1\}^m$, with $m=O(\log(n))$.
This result essentially generalizes and makes explicit a technique used implicitly in Ref.~\citep{kearns_learnability_1994}. In particular, the formulation we provide in Theorem~\ref{theo:warmup} makes clear (a) the applicability of the technique even in the case of \textit{noisy} random example access in the function case 
(i.e., when $\eta\neq 0$) and (b) that one can consider functions with up to logarithmically many output bits. Additionally, this result serves as a warm-up to familiarize the reader with the definitions from Section~\ref{sec:background}. 

We start with a lemma that shows the equivalence between the error (or ``loss'') in function learning and the error in distribution learning.
\begin{lemma}[Equivalence of distribution and function loss]
    \label{lem:pac_to_dtv}
    Let $f,h : \bsn \rightarrow \{0,1\}^m$ be two functions and $D_{f,P}$, $D_{h,P}$ be their two induced distributions. It holds for all distributions $P$ that
    \begin{align}
        \Prob_{x \sim P}[f(x) \neq h(x)] = d_{TV}(D_{f,P}, D_{h,P}) \text{.}
    \end{align}
\end{lemma}
\begin{proof}
We have
    \begin{align}
        d_{TV}(D_{f,P}, D_{h,P}) &= \frac{1}{2}\sum_{x\in\bsn}
        \underbrace{\sum_{y\in\{0,1\}^m} \left| D_{f,P}(x\|y) - D_{h,P}(x\|y) \right|}_{\begin{subarray}{1}=0 \text{, if } f(x)=h(x)\\ =2 \times P(x) \text{, if } f(x)\neq h(x)\end{subarray}}\\
        \nonumber
        &= \sum_{x\in\bsn}P(x)\times \left(1- \indexf(f(x),h(x)) \right)\\
        \nonumber
        &= \Prob_{x\sim P}[f(x) \neq h(x)].
    \end{align}
\end{proof}
To prove that hardness of learning a function class implies hardness of evaluator learning the induced distribution class, we will show that if we had an evaluator learner for some induced distribution class, then we can get a function learner for the underlying function class. To do this, we need a way to construct a function hypothesis from a given evaluator. A natural way to do this is to take a hypothesis that on any given input $x$, outputs a $y$, such that $x\|y$ is assigned the highest probability under the evaluator. 
More formally, if $D$ is some discrete probability distribution over $\{0,1\}^{n+m}$, then we let $h$
be defined via $h(x)=\arg \max_y D(x\|y)$. This argmax construction is a natural way to obtain a function hypothesis from an evaluator, and below we show that this is optimal.
\begin{lemma}[Optimal function hypothesis from an evaluator]
\label{lem:spiked_distr}
Let $D$ be some discrete probability distribution over $\{0,1\}^{n+m}$ and let $h$
be defined via $h(x)=\arg \max_y D(x\|y)$, for $x \in \bsn$, $y\in \{0,1\}^m$, then it holds for all functions $f:\bsn \rightarrow \{0,1\}^m$ and all probability distributions $P$ that 
\begin{align}
    {d_{TV}(D, D_{h,P}) \leq d_{TV}(D, D_{f,P})} \text{.}
\end{align}

\end{lemma}
\begin{proof}
  By definition, we have that
    \begin{align}
        d_{TV}(D,D_{h,P}) &= \frac{1}{2}\sum_x
        \sum_y\left|D(x\|y) - D_{h,P}(x\|y)\right| \\
        \nonumber
        &=\frac{1}{2}\sum_x\left(\sum_{y\neq h(x)}D(x\|y) + \left|D(x\|h(x))-P(x)\right|\right)\\
        \nonumber
        &:= \frac{1}{2}\sum_x L(x,h).
    \end{align}
    Now, we would like to show that $d_{TV}(D,D_{f,P}) - d_{TV}(D,D_{h,P}) \geq 0 $. Note, we have that
    \begin{align}
        d_{TV}(D,D_{f,P}) - d_{TV}(D,D_{h,P}) &= \frac{1}{2}\sum_x \left[ L(x,f) - L(x,h)\right]
    \end{align}
    and therefore, it is sufficient to show that $L(x,f) - L(x,h)\geq 0$ for all $x$. To this end,
    \begin{align}
        L(x,f) - L(x,h) &= \left[\left(\sum_{y\neq f(x),h(x)}D(x\|y)\right) + D(x\|h(x)) + \left|D(x\|f(x))-P(x)\right|\right] \\
        \nonumber
        &\qquad - \left[\left(\sum_{y\neq f(x),h(x)}D(x\|y)\right) + D(x\|f(x)) + \left|D(x\|h(x))-P(x)\right|\right]\\
        \nonumber
        &=  \left|D(x\|f(x))-P(x)\right| + D(x\|h(x))  - \left|D(x\|h(x))-P(x)\right| - D(x\|f(x)).
    \end{align}
    Thus, we need to show that
    \begin{align}
        \label{eq:1}
        \left|D(x\|f(x))-P(x)\right| + D(x\|h(x)) \geq \left|D(x\|h(x))-P(x)\right| - D(x\|f(x)).
    \end{align}
    Due to how $h$ is constructed, we have 
    \begin{align}
        D(x\|h(x)) \geq D(x\|f(x))
    \end{align}
    for all $x$ and can express $D(x\|h(x))=
    P(x)+\alpha$ and $D(x\|f(x))=P(x)+\beta$, where
    $1-P(x)\geq\alpha\geq\beta\geq -P(x)$. We plug this into the inequality \eqref{eq:1} and obtain
    \begin{align}
        \left|\beta\right| + \alpha \geq \alpha + \beta \text{.}
    \end{align}
\end{proof}
With these two lemmata in hand, we can prove the following result, which at a high level says that any PAC function learner for a given class of functions (mapping to at most logarithmically many output bits) can be turned into a evaluator learner for the induced distribution class, and vice versa.
\begin{theorem}
  \label{theo:warmup}
  Let $\mathcal{F}$ be some function class consisting only of functions $f : \bsn \rightarrow \{0,1\}^m$, for some $m=O(\log(n))$.
  Let $\mathcal{D}_{\mathcal{F},P,\eta}$ be the induced distribution class for some fixed probability distribution $P$
  over $\bsn$ and $0 \leq \eta < \frac{1}{2}$.
  \begin{itemize}
    \item
      If $\mathcal{D}_{\mathcal{F},P,\eta}$ is efficiently $(\epsilon, \delta)$-PAC
      EVAL learnable, then $\mathcal{F}$ is efficiently
      $(2(\eta+\epsilon), \delta, \REX(P, \eta))$-PAC learnable.
    \item
      If $\mathcal{F}$ is efficiently
      $(\epsilon, \delta, \REX(P, \eta))$-PAC learnable,
      then $\mathcal{D}_{\mathcal{F},P,\eta}$ is efficiently $(\eta+\epsilon, \delta)$-PAC EVAL learnable.
  \end{itemize}
\end{theorem}
\begin{proof}
  First statement:
  Assume $\AD$ is an efficient $(\epsilon, \delta)$-PAC EVAL learner for
  $\mathcal{D}_{\mathcal{F},P,\eta}$,
  then $\AD$ outputs an evaluator $\EVAL_{D'}$ to $D_{f,P,\eta}$ with $d_{TV}(D', D_{f,P,\eta}) \leq \epsilon$ for all $D_{f,P,\eta} \in \mathcal{D}_{\mathcal{F},P,\eta}$ with
  probability at least $1-\delta$.
  Note that $D'$ has no subscripts, as it is not necessarily an induced distribution and can have
  any structure. All we know is that $D'$ is $\epsilon$-close to $D_{f,P,\eta}$.
  We construct now the algorithm $\AF$, which simulates $\AD$
  by answering any sample accesses to $D_{f,P,\eta}$ with $\mathrm{query}[\REX(f,P,\eta)]$ and obtains $\EVAL_{D'}$.
  $\AF$ then outputs the function $h$ as an algorithm, which on input $x$
  calculates $p_i = \EVAL_{D'}(x\|y_i)$ for all possible $y_i \in \{0,1\}^m$ and
  returns the $y_i$ with the largest $p_i$. Since $m=O(\log(n))$, this is done in time $O(n)$.
  Due to Lemma~\ref{lem:spiked_distr}, Lemma~\ref{lem:pac_to_dtv} and the triangle inequality, we get
  \begin{align}
      \Prob_{x \sim P}\left[f(x)\neq h(x)\right] &= d_{TV}(D_{f,P}, D_{h,P})\\
      \nonumber
      &\leq d_{TV}(D_{f,P}, D_{h,P,\eta}) + d_{TV}(D_{f,P,\eta}, D') + d_{TV}(D',D_{h,P}) \\
      \nonumber
      &\leq \eta + \epsilon + d_{TV}(D', D_{f,P})\\
      \nonumber
      &\leq \eta + \epsilon + d_{TV}(D', D_{f,P,\eta}) + d_{TV}(D_{f,P,\eta}, D_{f,P})\\
      \nonumber
      &\leq 2(\eta+\epsilon) .
  \end{align}
  Thus, $\AF$ is an $(2(\epsilon+\eta), \delta, \REX(P,\eta))$-PAC learner for $\mathcal{F}$.

  Second statement:
  Assume $\mathcal{F}$ is efficiently $(\epsilon, \delta, \REX(P,\eta))$-PAC learnable,
  then there exists an algorithm $\AF$ that
  for all $f\in\mathcal{F}$, with probability $1-\delta$ and query access to $\REX(f, P,\eta)$ outputs a hypothesis $h\in\mathcal{F}$ where
  $\Prob_{x \sim P} \left[ f(x) \neq h(x) \right] \leq \epsilon$.
  We construct the learning algorithm $\AD$ to simulate
  $\AF$ (by answering any queries to the REX oracle with a sample from $D_{f,P,\eta}$) and output
  \begin{align}
    \EVAL_{D_{h,P}}(x\|y) =
    \begin{cases}
      P(x) \text{, if } h(x)=y \\
      0 \text{, else}
    \end{cases}.
  \end{align}
  Since
  \begin{align}
    d_{TV}(D_{h,P}, D_{f,P}) = \Prob_{x \sim P}\left[f(x)\neq h(x)\right] \leq \epsilon,
  \end{align}
  we have
  \begin{align}
    d_{TV}(D_{h,P}, D_{f,P,\eta}) &\leq d_{TV}(D_{h,P}, D_{f,P}) + d_{TV}(D_{f,P}, D_{f,P,\eta})\\
    \nonumber
    &\leq \eta+ \epsilon   \text{.}
  \end{align}
  Thus, $\EVAL_{D_{h,P}}$ is an $(\eta +\epsilon)$-close evaluator of $D_{f,P,\eta}$ and $\AD$ is an $(\eta+ \epsilon, \delta)$-PAC EVAL learner for
  $\mathcal{D}_{\mathcal{F},P,\eta}$.
\end{proof}
Theorem~\ref{theo:warmup} has been phrased in terms of efficient learnability -- i.e., it shows how efficient learners for one problem imply efficient learners for another. However, we can straightforwardly rephrase Theorem~\ref{theo:warmup} in terms of hardness implications. More specifically, we obtain from Theorem~\ref{theo:warmup} the following simple corollary: 

\begin{corollary}\label{cor:warmup} Let $\mathcal{F}$ be some function class consisting only of functions $f : \bsn \rightarrow \{0,1\}^m$, for some suitable ${m=O(\log(n))}$. Let $\mathcal{D}_{\mathcal{F},P,\eta}$ be the induced distribution class for some fixed probability distribution $P$
  over $\bsn$ and $0 \leq \eta < \frac{1}{2}$.
  \begin{itemize}
  \item if $\mathcal{D}_{\mathcal{F},P,\eta}$ is $(\eta + \epsilon,\delta)$-PAC $\EVAL$ hard, then $\mathcal{F}$ is
      $(\epsilon, \delta, \REX(P, \eta))$-PAC hard.
 \item if $\mathcal{F}$ is 
      $(2(\eta+\epsilon), \delta, \REX(P, \eta))$-PAC hard, then $\mathcal{D}_{\mathcal{F},P,\eta}$ is  $(\epsilon, \delta)$-PAC
      $\EVAL$ hard.
  \end{itemize}
\end{corollary}
As claimed in Table~\ref{table:overview}, the above corollary shows clearly that -- at least for the case of functions with at most logarithmically many output bits -- one can use function learning hardness results to prove distribution learning hardness results, and vice versa. We note that this correspondence holds even for function learning with \textit{noisy} random examples, provided one considers the appropriate associated induced distribution class. Additionally, we note that one can also straightforwardly extend Theorem~\ref{theo:warmup} to the setting in which both the function and distribution learner have \textit{statistical query} access, as opposed to (noisy) random example and sample access as considered here.


\subsection{Evaluator learning hardness from weak PRFs}
\label{subsec:classhard}

In the previous section, we layed out how one can obtain EVAL learning hardness from PAC-hard functions by using the construction of induced distributions, for functions that map to $m=O(\log(n))$ bits.
We are now interested in whether one can obtain EVAL learning hardness in a more general setting, or from other primitives.
We will see in this section that we can indeed obtain EVAL learning hardness by using weak-secure PRFs as the distribution inducing functions. We note that this is very closely related to prior work in Ref.~\cite{kearns_learnability_1994}, where it was implicitly shown that one can use \textit{classic-secure} PRFs to obtain \textit{generator} learning hardness. This was made explicit and generalized in Ref~\cite{sweke_generator_2021}, which used this technique to prove a quantum-classical separation for generative modelling. However, in Ref.~\cite{sweke_generator_2021} it was posed as an open question whether or not one can obtain distribution learning hardness results from \textit{weak} PRFs, and it is this question which we answer in the affirmative here, for the case of \textit{evaluator} learning. Apart from allowing us to prove the quantum-classical distribution learning separation in Section~\ref{subsec:quanteasy}, this result also opens the possibility of proving classical hardness results based on weaker assumptions than those necessary for candidate classic-secure PRFs, or hard to learn function classes.

\begin{figure}
\includegraphics[scale=0.4]{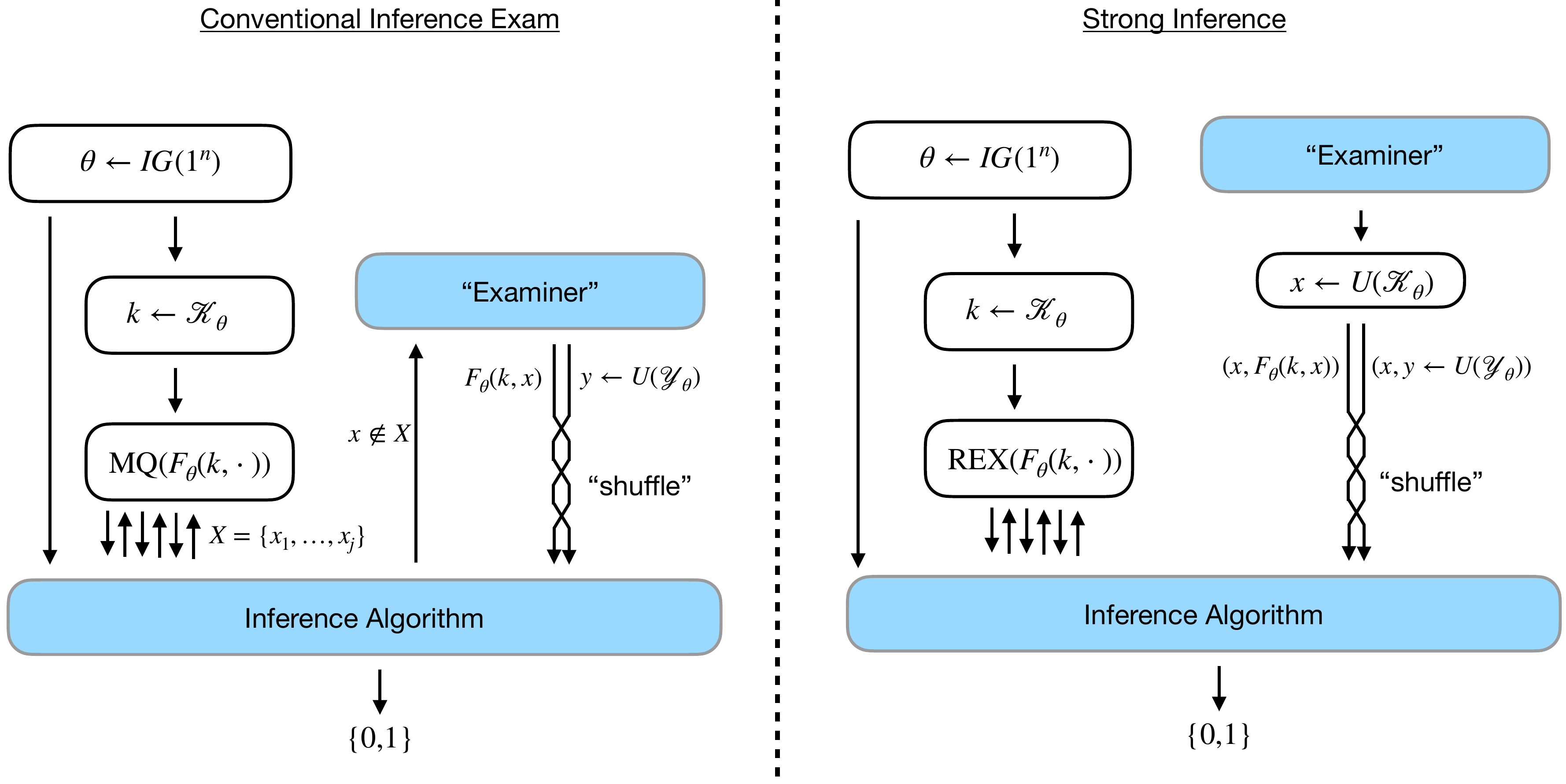}
\caption{Illustration of the difference between a conventional inference exam and a strong inference exam. Conventional inference exams have been introduced in  Ref.~\cite{goldreich_randomfunctions_1986} as an alternative characterization of PRFs, and we note here that strong inference exams can be used for the characterization of weak PRFs.}
\label{fig:inference}
\end{figure}

To this end, we note that there is a very useful characterization of classic-secure PRFs, namely the ability to withstand a so-called \textit{inference exam} \cite{goldreich_randomfunctions_1986}. In particular, this characterization has been crucial for proving generator learning hardness for distributions induced by classic-secure PRFs~\cite{kearns_learnability_1994, sweke_generator_2021}. As shown in Figure~\ref{fig:inference}
the inference exam is a procedure where a learner is tested whether it has really learned anything about the function by asking whether it can distinguish between a function input/output pair and a random input/output pair.
It has been shown that this distinguishing task is not
possible efficiently, with convincing probability, if the function is a classic-secure PRF \citep{goldreich_randomfunctions_1986}. For our proof of  EVAL learning hardness from weak-secure PRFs we will make use of a similar characterization of weak PRFs in terms of a slightly modified inference exam. More specifically, as shown in Figure~\ref{fig:inference}, in a conventional inference exam (defined in Ref.~\cite{goldreich_randomfunctions_1986}), the distinguishing algorithm has (a) \textit{membership query} access to the unknown function, and (b) the ability to choose its own exam string. We will define a \textit{strong} inference exam, in which the distinguishing algorithm has (a) only random example access to the unknown function, and (b) gets given an exam pair drawn at random. Intuitively, we call this a \textit{strong} inference exam as it is harder to pass than the conventional inference exam, due to the weaker oracle access and lack of ability to choose the exam. As we will see, we can characterize weak-PRFs in terms of the existence or non-existence of algorithms which succeed in the task of strong-inference, analogously to how classic-secure PRFs can be characterized in terms of inference exams.

Let us now describe the strong inference exam, that is used to characterize weak-secure PRFs and we refer the interested reader to Ref.~\cite{goldreich_randomfunctions_1986, sweke_generator_2021} for the details of standard inference exams. In the following we consider function collections $\KPRFC$ which have the structure of PRF collections, where $\mathcal{K}_\theta$ is the so-called secret key space. In particular, let $k \in \mathcal{K}_\theta$ be some fixed but unknown secret key. We then define a strong inference exam as follows:
\begin{definition}[Strong inference exam]
    \label{def:infexam}
    Let $\mathcal{A}$ be some probabilistic polynomial time classical algorithm that ``takes the exam''.
    On input $\theta \in \Theta$, $\mathcal{A}$ can carry out any computation while having access to $\REX(F_{\theta}(k,\cdot), U)$.
    After some time, when $\mathcal{A}$ signals that it is ready for the exam, $\mathcal{A}$ is presented two pairs $(x', f_1)$ and $(x', f_2)$ in random order, where $x' \sim U(\mathcal{X}_\theta)$, $f_1 = F_{\theta}(k, x')$ and $f_2 \sim U(\mathcal{Y}_\theta \setminus \{F_\theta(k, x')\})$. We say that $\mathcal{A}$ 
    ``passes the exam'' 
    if it correctly guesses which of the two pairs
    stems from the function $F_\theta$ and which was the random value from $U(\mathcal{Y}_\theta)$.
\end{definition}

Analogously to the definition of $Q$-inference in 
Ref.~\citep{goldreich_randomfunctions_1986}, we now define the notion of \textit{strong} Q-inference, which defines a lower bound on the probability of $\mathcal{A}$ passing the random example inference exam.

\begin{definition}[Strong Q-inference]
    \label{def:qinf}
    Let $Q$ be some function. We say that $\mathcal{A}$ strongly $Q$-infers the collection
    $\KPRFC$ if for infinitely many $n$, given input $\theta \in \Theta$, it holds
    that
    \begin{align}
        \Prob[\text{"$\mathcal{A}$ passes the exam"}] \geq 1/2 + 1/Q(n) \text{,}
    \end{align}
    where the probability is taken uniformly over all possible choices of
    $\theta\leftarrow\mathcal{IG}(1^n), k \in \mathcal{K}_{\theta}$, $x' \in \mathcal{X}_\theta$,
    $f_2 \in \mathcal{Y}_{\theta}$, and all possible orderings of the exam pairs.
    We say that a function collection can be polynomially strongly inferred if there exists a
    polynomial $Q$ and a probabilistic polynomial time algorithm $\mathcal{A}$ which strongly
    $Q$-infers the collection.
\end{definition}

We are now ready to state a core lemma, that is analogous to the result in Ref.~\citep{goldreich_randomfunctions_1986}, which characterizes PRFs via the strong polynomial Q-inference exam.

\begin{lemma}[Weak-secure PRFs cannot be polynomially strongly inferred]
    \label{lemma:qinf_prf}
    Let $\mathcal{F} = \KPRFC$ be a collection of efficiently computable functions.
    $\mathcal{F}$ is weak-secure pseudo-random if and only if $\mathcal{F}$ cannot be polynomially strongly inferred.
\end{lemma}
\begin{proof}
The proof is a direct generalization of the proof for Theorem 4 in 
Ref.~\citep{goldreich_randomfunctions_1986}.
\end{proof}

With the necessary definitions at hand, we can now present Theorem~\ref{theo:classical_eval_hardness}, which shows that
distributions induced by weak-secure pseudo-random functions are hard to EVAL learn. \newline


\begin{theorem} [Classical evaluator learning hardness from weak pseudo-random functions] \label{theo:classical_eval_hardness} Let 
    \begin{align}\mathcal{F} = \left\{F_{\theta} : \mathcal{K}_\theta \times \mathcal{X}_{\theta} \rightarrow \mathcal{Y}_{\theta} | \theta \in \Theta\right\}
    \end{align}
    be a weak-secure pseudo-random function collection, where for all $\theta$ one has that $\mathcal{X}_\theta = \mathcal{Y}_\theta = \bsn$ for some $n$. For all $\theta \in \Theta$ and all $k \in \mathcal{K}_\theta$, we define the induced probability distribution $\DPK :=D_{F_{\theta}(k,\cdot), U}$, and the associated distribution class $\CN:=\left\{\DPK \mid \theta \in \Theta, k \in \mathcal{K}_{\theta}\right\}$.
    For all sufficiently large $n$,  all $\epsilon < 1/9$ and $\delta \leq 1/5-\Omega(1/\mathrm{poly}(n))$,  there exists no efficient classical $(\epsilon, \delta)$-PAC
    EVAL learner of $\CN$.
\end{theorem}

In order to prove Theorem  \ref{theo:classical_eval_hardness}, we require the following two technical lemmas. 

\begin{lemma}
\label{lem:alpha1}
Let $f:\{0,1\}^n\rightarrow\{0,1\}^n$ and let $D'$ be some distribution satisfying $d_{TV}(D_{f,U},D')\leq \epsilon$, for some $\epsilon < 1/9$. Then, for at least $3/4 \times2^n$ strings $x$, it holds that 
\begin{align}
    D'(x\| f(x)) \geq \frac{\epsilon}{2^n}.
\end{align}
\end{lemma}
\begin{proof}
Per contradiction, assume that the claim is false, thus, it holds that for at least 
$1/4 \times 2^n$ strings $x$, 
\begin{align}
D'(x\| f(x)) < \frac{\epsilon}{2^n}. 
\end{align}
It follows that 
\begin{align}
     d_{TV}(D_{f,U},D') &= \frac{1}{2} \sum_{x,y \in \{0,1\}^{n}} {\left| D_{f,U}(x\|y) - D'(x\|y) \right|}\\
     \nonumber
    &= \frac{1}{2} \sum_{x\in \{0,1\}^n}{\left| \frac{1}{2^n} - D'(x\| f(x)) \right|} + \sum_{\text{other } x,y} {\left| D'(x\|y) \right|}\\
    \nonumber
    &\geq \frac{1}{2} \sum_{x\in \{0,1\}^n}{\left| \frac{1}{2^n} - D'(x\| f(x)) \right|}\\
    \nonumber
    &> \frac{1}{2} \left( \frac{2^n}{4} \right) \left[ \frac{1}{2^n} - \frac{\epsilon}{2^n} \right]\\ 
    \nonumber
    &= \frac{1}{8} (1-\epsilon)\\
    \nonumber
    &> \epsilon \qquad\qquad\qquad\qquad\qquad\left(\text{when } \epsilon < \frac{1}{9} \right)
\end{align}
which contradicts the assumption.
\end{proof}

\begin{lemma}
    \label{lem:alpha2}
    Let $f:\{0,1\}^n\rightarrow\{0,1\}^n$ and let $D'$ be some distribution satisfying $d_{TV}(D_{f,U},D')\leq \epsilon$. For at least $\frac{1}{2}\times \left( 2^{2n}-2^n \right)$ of the strings $x\|y$ with $y\neq f(x)$, it holds that 
    \begin{align}
    D'(x\|y) \leq \frac{4\epsilon}{2^{2n}-2^n}.
    \end{align}
\end{lemma}
\begin{proof}
Per contradiction, assume that the claim is false, and therefore, for at least $\frac{1}{2}\times \left( 2^{2n}-2^n \right)$ of the strings $x\|y$ with $y\neq f(x)$, it holds that
\begin{align}
    D'(x\| y) > \frac{4\epsilon}{2^{2n}-2^n}. 
\end{align}
From this, it follows that 
\begin{align}
      d_{TV}(D_{f,U}, D') &= \frac{1}{2} \sum_{x,y \in \{0,1\}^{n}} {\left| D_{f,U}(x\|y) - D'(x\|y) \right|}\\
      \nonumber
    &= \frac{1}{2} \sum_{x\in \{0,1\}^n}{\left| \frac{1}{2^n} - D'(x\| f(x))  \right|} + \sum_{\text{other } x,y} {\left| D'(x\|y) \right|}\\
    \nonumber
    &\geq \frac{1}{2} \underbrace{ \sum_{\text{other } x,y} {\left| D'(x\|y) \right|} }_{(2^{2n}-2^n)\text{ many}}\\
    &> \frac{1}{2} \times \frac{1}{2} \left( 2^{2n}-2^n \right) \left[ \frac{4\epsilon}{2^{2n}-2^n} \right]\\
    \nonumber
    &= \epsilon,
\end{align}
which contradicts the assumption.
\end{proof}
Given the above lemmas, we can now prove Theorem \ref{theo:classical_eval_hardness}.
\begin{proof}[Proof for Theorem  \ref{theo:classical_eval_hardness}] The proof will be by contradiction. To do this, we assume that, for some $\epsilon<1/9$ and ${\delta \leq 1/5-\Omega(1/\mathrm{poly}(n))}$, there exists a classical efficient $(\epsilon,\delta)$-PAC evaluator learner of $\CN$ - i.e., a polynomial time probabilistic classical algorithm $\tilde{\mathcal{A}}$, which
for all $\DPK \in \CN$, when given sample access to $\DPK$, outputs with probability at least
$1-\delta$, an evaluator $\EVAL_{D'}$ for some distribution $D'$ satisfying $d_{TV}(\DPK,D')\leq \epsilon$. We now use this assumption to construct an efficient classical algorithm $\mathcal{A}$ that uses $\tilde{\mathcal{A}}$ to polynomially strongly infer $\mathcal{F}$. This strong polynomial-inference is per definition of $\mathcal{F}$ not possible, resulting in the
sought contradiction.

So, let us describe algorithm $\mathcal{A}$:
When given access to $\REX(F_\theta(k,\cdot), U)$, algorithm $\mathcal{A}$ starts by simulating algorithm $\tilde{\mathcal{A}}$. In particular, for every query made by $\tilde{\mathcal{A}}$, algorithm $\mathcal{A}$ queries $\REX(F_\theta(k,\cdot), U)$, obtains some tuple $(x,F_\theta(k, x))$, and then passes the string $x||F_\theta(k, x)$ to $\tilde{\mathcal{A}}$. As this is indistinguishable from a sample query to the distribution $D_{(\theta,k)}$, algorithm $\tilde{\mathcal{A}}$ will, after a polynomial number of queries, output with probability at least $1-\delta$, an evaluator $\EVAL_{D'}$ for some distribution $D'$ satisfying $d_{TV}(\DPK,D')\leq \epsilon$.

At this stage, $\mathcal{A}$ is ready to take the random example inference exam, and when presented the
two exam pairs $s_1=(x', f_1)$ and  $s_2=(x', f_2)$,
$\mathcal{A}$ will run the strategy presented in  Algorithm \ref{alg:evalqinfer} to determine which of
the two values $f_1,f_2$ is $F_\theta(k, x')$.

\begin{algorithm}
    \caption{Exam strategy to infer $\mathcal{F}$}\label{alg:evalqinfer}
    \SetKwInOut{Input}{Input}
    \SetKwInOut{Output}{Output}
    \Input{Exam strings: $s_1, s_2 \in \{0,1\}^{2n}$, Evaluator: $\EVAL_{D'}$, parameters: $n$, $\epsilon$}
    \Output{Index of the exam string that has been produced by $F_\theta(k, x')$}
    \nl $p_1 \leftarrow \EVAL_{D'}(s_1)$\;
    \nl $p_2 \leftarrow \EVAL_{D'}(s_2)$\;
    \nl \uIf{$p_1 \geq \frac{\epsilon}{2^n}$ and $p_2 \leq \frac{4\epsilon}{2^{2n}-2^n}$}{
        \nl return 1\;
    }
    \nl \uElseIf{$p_1 \leq \frac{4\epsilon}{2^{2n}-2^n}$ and $p_2  \geq \frac{\epsilon}{2^n}$}{
        \nl return 2\;
    }
    \nl \Else{
        \nl return $random \sim U(\{1,2\})$\;
    }
\end{algorithm}

We will now analyse the probability that $\mathcal{A}$ passes the random example inference
 exam.
 Firstly, note that for $n \geq 3$, we have that $\epsilon/2^n > 4\epsilon/(2^{2n}-2^n)$ and
 thus the conditions in lines \circled{3} and \circled{5} of Algorithm \ref{alg:evalqinfer} cannot be true at the same time.
Additionally, if algorithm $\tilde{\mathcal{A}}$ was successful, which happens with probability at least $1-\delta$, then it follows from Lemma \ref{lem:alpha1} and \ref{lem:alpha2}, as well as the promise that $\epsilon < 1/9$, that
 
\begin{align}
     \underset{x' \sim U(\mathcal{X}_\theta)}{\Prob}\left[\EVAL_{D'}(x' \| F_\theta(k, x')) \geq \frac{\epsilon}{2^n}\right] &\geq \frac{3}{4}  ,\\
     \underset{x' \sim U(\mathcal{X}_\theta)}{\Prob}\left[\EVAL_{D'}(x' \| F_\theta(k, x')) < \frac{\epsilon}{2^n}\right]  &< \frac{1}{4}  ,\\
     \underset{\begin{substack} {x' \sim U(\mathcal{X}_\theta) \\ u\sim U(\mathcal{Y}_\theta \setminus \{F_\theta(k, x')\})} \end{substack}}{\Prob}\left[\EVAL_{D'}(x' \| u) \leq \frac{4\epsilon}{2^{2n}-2^n}\right] & \geq \frac{1}{2} ,\\
     \underset{ \begin{substack} {x' \sim U(\mathcal{X}_\theta) \\ u\sim U(\mathcal{Y}_\theta \setminus \{F_\theta(k, x')\})} \end{substack} }{\Prob}\left[\EVAL_{D'}(x' \| u) > \frac{4\epsilon}{2^{2n}-2^n}\right] &< \frac{1}{2}.
\end{align}
 From the above, we can now bound the probability that algorithm $\mathcal{A}$ is successful, conditioned on $\tilde{\mathcal{A}}$ being successful. To do this, we denote the event that "Algorithm \ref{alg:evalqinfer} returns on line \circled{$l$} given that $\tilde{\mathcal{A}}$ was successful" by "\circled{$l$}". Using this, we have the following.\newline

 \noindent \textbf{Case 1.} $s_1 = x' \| F_\theta(k, x')$ and $s_2 = x' \| u$:\\
 \begin{align}
     &\Prob\left[\text{\circled{4}}\right] \geq \frac{1}{2} \times \frac{3}{4} =\frac{3}{8},\\
     &\Prob\left[\text{\circled{6}}\right] < \frac{1}{2} \times \frac{1}{4} =\frac{1}{8} ,
 \end{align}
 and therefore
 \begin{align}
 \Prob\left[ \text{Algorithm \ref{alg:evalqinfer} returns 1} \mid \tilde{\mathcal{A}}\text{ successful} \right]
      &=  \Prob\left[\text{\circled{4}}\right] + \frac{1}{2} \Prob\left[\text{\circled{8}}
      \right] \\
      \nonumber
      &=\Prob\left[\text{\circled{4}}\right] + \frac{1}{2}\left[ 1 - \left( \Prob\left[\text{\circled{4}}\right] + \Prob\left[\text{\circled{6}}\right] \right)\right] \\
      \nonumber
      &= \frac{1}{2} + \frac{1}{2} \Prob\left[\text{\circled{4}}\right] - \frac{1}{2} \Prob\left[\text{\circled{6}}\right]
      ,\\
      \nonumber
      &\geq \frac{5}{8}.
 \end{align}

\noindent \textbf{Case 2.} $s_1 = x' \| u$ and $s_2 = x' \| F_\theta(k, x')$:\\
 \begin{align}
     &\Prob\left[\text{\circled{4}}\right] < \frac{1}{2} \times \frac{1}{4} =\frac{1}{8}, \\
     &\Prob\left[\text{\circled{6}}\right] \geq \frac{1}{2} \times \frac{3}{4} =\frac{3}{8},
 \end{align}
  and therefore
 \begin{align}
 \Prob\left[ \text{Algorithm \ref{alg:evalqinfer} returns 2} \mid \tilde{\mathcal{A}}\text{ successful} \right]
      &=  \Prob\left[\text{\circled{6}}\right] + \frac{1}{2} \Prob\left[\text{\circled{8}}\right] \\
      \nonumber
      &=\Prob\left[\text{\circled{6}}\right] + \frac{1}{2}\left[ 1 - \left( \Prob\left[\text{\circled{6}}\right] + \Prob\left[\text{\circled{4}}\right] \right)\right] \\
      \nonumber
      &= \frac{1}{2} + \frac{1}{2} \Prob\left[\text{\circled{6}}\right] - \frac{1}{2} \Prob\left[\text{\circled{4}}\right]\\
      \nonumber
      &\geq \frac{5}{8}.
 \end{align}
 
\noindent Thus, taking both cases together, we have that
\begin{equation}
\Prob\left[ \text{Algorithm \ref{alg:evalqinfer} returns correct index} \mid \tilde{\mathcal{A}}\text{ successful}\right] \geq \frac{5}{8}.
\end{equation}
Using the above, we can now lower bound the probability that $\mathcal{A}$ passes the exam,
to get
\begin{align}
\Prob\left[ \mathcal{A} \text{ passes exam}\right] &= \Prob\left[ \text{Alg. \ref{alg:evalqinfer} returns correct index} \mid \tilde{\mathcal{A}}\text{ successful}\right]\times \mathrm{Pr}[\tilde{\mathcal{A}}\text{ successful}] \\
\nonumber
     &\qquad\qquad+ \Prob\left[ \text{Alg. \ref{alg:evalqinfer} returns correct index} \mid \tilde{\mathcal{A}}\text{ unsuccessful}\right]\times \mathrm{Pr}[\tilde{\mathcal{A}}\text{ unsuccessful}]\\
     \nonumber
     &\geq \Prob\left[ \text{Alg. \ref{alg:evalqinfer} returns correct index} \mid \tilde{\mathcal{A}}\text{ successful}\right]\times \mathrm{Pr}[\tilde{\mathcal{A}}\text{ successful}] \\
     \nonumber
     &\geq \frac{5}{8}(1-\delta)\\
     \nonumber
     &=\frac{5}{8} - \frac{5}{8}\delta \\
     \nonumber
     &\geq \frac{5}{8} - \frac{5}{8}\left[\frac{1}{5}-\Omega\left(\frac{1}{\mathrm{poly}(n)}\right)\right]\\
     \nonumber
     &= \frac{1}{2} + \Omega\left(\frac{1}{\mathrm{poly}(n)}\right).
\end{align}
Therefore, $\mathcal{A}$ polynomially strongly infers $\mathcal{F}$, which contradicts the assumption that $\mathcal{F}$ is weak-secure pseudo-random.

\end{proof}

\section{A quantum-classical separation for density modelling} \label{subsec:quanteasy}

In this work, we are interested in obtaining a quantum-classical separation for density modelling (evaluator learning).
So far, we have seen in
Section \ref{subsec:classhard} that when we instantiate the ``function to distribution'' construction with a weak-secure PRF, we can achieve a classical hardness result. 
The question is therefore, is there a weak-secure PRF which allows us to also prove an efficient quantum learning result?
In this section we show that by using the PRF previously utilized in 
Ref.~\citep{sweke_generator_2021} to show a separation for GEN learning, we can also achieve a separation for EVAL learning.

To begin, we restate the definition of the PRF used in 
Ref.~\citep{sweke_generator_2021}, and we refer there for additional details and discussion. Let $p \in \mathbb{N}$, we say that an element $y \in \Zq$ is a quadratic residue modulo $p$ if there exists an $x\in \Zq$ such that $x^2 \equiv y \mod p$.
Additionally, we say that $p$ is a safe prime if  $p = 2q +1$ with $q$ prime.
Let $\QRp$ be the set of quadratic residues modulo $p$ and $\left\{ \QRp \right\}$ be the set of such sets where $p$ is a safe prime.
Define the parameter set $\mathcal{P}_{(p,g,g^a)}$ as the infinite set of all tuples of the form $(p,g,g^a)$
where $p$ is some safe prime, $g$ is a generator for $\QRp$ and $a \in \Zq$.
We denote the subset of all such tuples in which $p$ is an $n$-bit prime as $\mathcal{P}_{n, (p,g,g^a)}$ and we note that $\mathcal{P}_{(p,g,g^a)} = \bigcup_{n\in \mathbb{N}} \mathcal{P}_{n,(p,g,g^a)}$ is an efficiently sampleable parameter set (see Ref.~\cite{sweke_generator_2021} for a description of the efficient instance generation algorithm).
Now, given some safe prime $p=2q+1$, define the function $f_p : \QRp \rightarrow \Zq$ via
\begin{align}
    f_p(x)= \begin{cases}x & \text { if } x \leq q \\ p-x & \text { if } x>q\end{cases}.
\end{align}
This allows us to define the functions $\tilde{G}_{(p, g, g^a)}^0$ and  $\tilde{G}_{(p, g, g^a)}^1 : \Zq \rightarrow \Zq$ via
\begin{align}
    \tilde{G}_{(p, g, g^a)}^0(b) &:= f_p(g^a \mod p), \\
    \tilde{G}_{(p, g, g^a)}^1(b) &:= f_p(g^{a b} \mod p).
\end{align}
With this in hand, we can finally define the  function collection $\left\{ \Fpgga \mid (p,g,g^a) \in \mathcal{P}_{p,g,g^a} \right\}$, where
\begin{align}
    \Fpgga : \Zq \times \bsn \rightarrow \Zq
\end{align}
is defined algorithmically in Algorithm~\ref{alg:fcompute} below:
\begin{algorithm}
    \caption{Algorithmic implementation of $F_{(p,g,g^a)}$}\label{alg:fcompute}
    \SetKwInOut{Input}{Input}
    \SetKwInOut{Output}{Output}
    \Input{Function input: $x \in \bsn$, Secret key: $k \in \Zq$, Parameters: $p, g, g^a$}
    \Output{The function value of $F_{(p,g,g^a)}(k, x)$}
    \nl $b_0 \leftarrow k$\;
    \nl\For{$1 \leq j \leq n$}{
        \nl \uIf{$x_j = 0$}{
            \nl $b_j \leftarrow \tilde{G}_{(p, g, g^a)}^0(b_{j-1})$;
        }
        \nl \uElseIf{$x_j = 1$}{
            \nl $b_j \leftarrow \tilde{G}_{(p, g, g^a)}^1(b_{j-1})$;
        }
    }
    \nl return $b_n$;
\end{algorithm}

As shown in Ref.~\cite{sweke_generator_2021} the function collection $\left\{ \Fpgga \mid (p,g,g^a) \in \mathcal{P}_{p,g,g^a} \right\}$ is a 
classic-secure PRF collection, under the Decisional Diffie-Hellman (DDH) assumption\footnote{A reader familiar with pseudorandom functions may recognize Algorithm~\ref{alg:fcompute} as the \emph{Goldreich-Goldwasser-Micali} construction of a PRF, from the pseudorandom generator implicit in the \emph{Diffie-Hellman} assumption.}. Given that any classic-secure PRF is also a weak-PRF, we know (from Theorem~\ref{theo:classical_eval_hardness}) that we can instantiate the "function to distribution" construction with $\{F_{(p,g,g^a)}(k, \cdot)\,|\, (p,g,g^a)\in\mathcal{P}_{p,g,g^a}, k\in \Zq\}$ to obtain a distribution class $\mathcal{D} = \{D_{(p,g,g^a),k}\,|\,(p,g,g^a)\in\mathcal{P}_{p,g,g^a}, k\in \Zq\}$ which is hard to evaluator learn. However, as done 
in Ref.~\cite{sweke_generator_2021} we will in fact consider a slight modification of the induced distribution class -- for reasons which will soon become clear -- in which an encoding of the parameters $(p,g,g^a)$ is appended onto the samples. More specifically, recall that one samples from $D_{(p,g,g^a),k}$ by first drawing $x\leftarrow \{0,1\}^n$, and then outputting $x\|F_{(p,g,g^a)}(k,x)$. We define the distribution $\tilde{D}_{(p,g,g^a),k}$ as the distribution which is sampled from by first drawing $x\leftarrow \{0,1\}^n$, and then outputting $x\|F_{(p,g,g^a)}(k,x)\|(p,g,g^a)$ - i.e., the exact same distribution as $D_{(p,g,g^a),k}$, but with the parameters appended to each sample. Naturally, we then define the distribution class
\begin{equation}
\tilde{\mathcal{D}} = \{\tilde{D}_{(p,g,g^a),k}\,|\,(p,g,g^a)\in\mathcal{P}_{p,g,g^a}, k\in \Zq\}.
\end{equation}
As discussed in Ref.~\cite{sweke_generator_2021}, given the fact that any candidate inference algorithm for a PRF (or weak PRF) is also given the parameters of the unknown PRF (see Figure~\ref{fig:inference}), the proof of Theorem~\ref{theo:classical_eval_hardness} is unaffected if one uses the distribution class $\tilde{\mathcal{D}}$ in place of $\mathcal{D}$. As such, we obtain the following Corollary from Theorem~\ref{theo:classical_eval_hardness}, and the fact that $\left\{ \Fpgga \mid (p,g,g^a) \in \mathcal{P}_{p,g,g^a} \right\}$ is a classic-secure PRF collection under the DDH assumption: 


\begin{corollary}[Evaluator learning hardness of $\tilde{\mathcal{D}}$]
    \label{cor:ddh_sep}
    Under the Decisional Diffie Hellman assumption, for all sufficiently large $n$, all $\epsilon< 1/9$ and all $\delta\leq 1/5 - \Omega(1/\mathrm{poly}(n))$, there is no efficient classical $(\epsilon,\delta)$-PAC EVAL learner for $\tilde{\mathcal{D}}$.
\end{corollary}

We would now like to show that one can indeed obtain an efficient quantum evaluator learner for $\tilde{\mathcal{D}}$. To this end, we start with the following observation. 

\begin{observation} [Exact evaluator from knowing the secret key]
    For all $\tilde{D}_{(p,g,g^a),k} \in \tilde{\mathcal{D}}$, given the secret key $k$, along with parameters $(p,g,g^a)$, one can output an efficient exact evaluator of $\tilde{D}_{(p,g,g^a),k}$.
\end{observation}
The above observation can be easily understood by the considering the following evaluator
\begin{align}
    \EVAL_{\tilde{D}_{(p,g,g^a),k}}(x\|y\|(p,g,g^a)) =
    \begin{cases}
      \frac{1}{2^n}, &\text{ if } F_{(p,g,g^a)}(k, x)=y\\
      0, &\text{ else}
    \end{cases}.
\end{align}
It follows from the construction of $F_{(p,g,g^a)}$ that the evaluator $\EVAL_{\tilde{D}_{(p,g,g^a),k}}$ is computable in poly-time if both $k$, as well as $(p,g,g^a)$ are known.
In light of this, we see that learning an evaluator for $\tilde{D}_{(p,g,g^a),k}$ reduces to learning, from samples, the parameters $(p,g,g^a)$ as well as the secret key $k$. However, by design, the parameters $(p,g,g^a)$ come ``for free'' with each sample, and therefore, one only needs to learn the secret key $k$. To this end, we note that precisely such an algorithm has already been constructed in Ref.~\cite{sweke_generator_2021}. More specifically, Ref~\cite{sweke_generator_2021} has constructed an efficient quantum algorithm which, by using the exact quantum algorithm for discrete logarithms as a subroutine, can deterministically recover the secret key $k$ from samples from $\tilde{D}_{(p,g,g^a),k}$. Putting it all together, we see that by using the efficient (deterministic) quantum key-learning algorithm from Ref.~\cite{sweke_generator_2021}, coupled with fact that the parameters $(p,g,g^a)$ are given for free, and that together the key $k$ and parameters $(p,g,g^a)$ fully specify an (exact) evaluator for $\tilde{D}_{(p,g,g^a),k}$, we obtain the following corollary:

\begin{corollary} \label{cor:qddh}
    $\tilde{\mathcal{D}}$ is quantumly efficiently $(\epsilon=0, \delta=0)$-PAC evaluator learnable.
\end{corollary}
Juxtaposing Corollary~\ref{cor:ddh_sep} with Corollary~\ref{cor:qddh} yields a super-polynomial quantum-classical separation for density modelling, up to the \emph{Decisional Diffie-Hellman} assumption.

\section{Conclusions and outlook} \label{sec:discussion}


In this work, we have provided a variety of rigorous insights into the relative power of classical and quantum computers for the task of density modelling.
Specifically, we first provided an overview of techniques for proving distribution learning hardness from various classes of functions.
Apart from providing a comprehensive picture of existing techniques, we have (a) provided a generalization of methods for proving distribution learning hardness from PAC hard-to-learn functions and (b) shown that weak-secure PRFs are sufficient to prove hardness of evaluator learning.
Given this, we have then shown that there exists a density modelling task which is provably hard for classical computers, but can be solved by an efficient quantum learning algorithm. This separation contributes to the relatively scarce collection of machine learning type problems for which one can rigorously prove a quantum advantage~\cite{onseparations}. 
In our outlook, we like to formulate the following open research questions:
\begin{enumerate}
    \item Can one get a computational separation (possibly with a fault-tolerant quantum computer) for a realistic       learning task?
   Indeed, the learning task considered in this work involved a synthetic and highly fine-tuned distribution that almost certainly does not appear naturally and is not of any practical relevance. As such, it is still an open question whether one can find a practical -- or ``real world'' -- learning task for which quantum computers offer a super-polynomial speedup.
    
    \item Furthermore, a major question is whether (even for learning problems that are synthetic and not of any practical relevance) one can prove a quantum advantage using a quantum algorithm that works on noisy, near-term quantum devices instead of large-scale error-corrected quantum computers. Indeed, it is the hope that by formalizing and abstracting methods for proving classical hardness results in distribution learning, this work stimulates  and facilitates such research efforts.
        
    \item To that end, it is an interesting research question to find a quantum-classical
        learning separation based on weak- but not classic-secure PRFs. More specifically, a separation which requires a weaker assumption than that necessary for the existence of classic-secure PRFs. The hope is that whichever assumption is used for classical hardness, can be broken by \textit{near-term} quantum devices. 
        What first comes to mind when pursuing this idea, is to use weak-secure PRFs based
        on the hardness of \emph{learning parity with noise} (LPN) \citep{bogdanov_pseudo-random_2017}. In particular, while learning such a PRF classically is believed to be hard, there are efficient quantum learning 
        algorithms \citep{cross_quantumlpn_2015}. However, these quantum algorithms require access to a \textit{quantum} random example oracle, and it is not clear how to overcome this limitation.
        Indeed, it is an interesting question of independent interest whether there exist candidate weak-PRFs, which are not secure against near-term quantum adversaries.
        
    \item Finally, the answers to the open questions in Table \ref{table:overview} are certainly interesting and important.
\end{enumerate}

\section*{Acknowledgements}

We would like to thank Thomas Vidick for discussions. R.~S.~is very grateful to Alex Nietner, Marcel Hinsche and Marios Ioannou for many discussions and insights into both quantum and classical distribution learning.
The authors acknowledge partial funding by the Einstein Research Unit ``Perspectives of a quantum digital transformation: Near-term quantum computational devices and quantum processors'' of the Berlin University Alliance. J.~E.~and R.~S.~have also received funding from the MATH+ Cluster of Excellence, the BMWK (PlanQK, 
EniQmA), the BMBF (Hybrid), and the QuantERA (HQCC). This research is also part of the Munich Quantum Valley (K8), which is supported by the Bavarian state government with funds from the Hightech Agenda Bayern Plus. 
J.-P.~S.~received funding from the Berlin Institute for the Foundations of Learning and Data (BIFOLD).


\begin{thebibliography}{19}
\expandafter\ifx\csname natexlab\endcsname\relax\def\natexlab#1{#1}\fi
\expandafter\ifx\csname bibnamefont\endcsname\relax
  \def\bibnamefont#1{#1}\fi
\expandafter\ifx\csname bibfnamefont\endcsname\relax
  \def\bibfnamefont#1{#1}\fi
\expandafter\ifx\csname citenamefont\endcsname\relax
  \def\citenamefont#1{#1}\fi
\expandafter\ifx\csname url\endcsname\relax
  \def\url#1{\texttt{#1}}\fi
\expandafter\ifx\csname urlprefix\endcsname\relax\def\urlprefix{URL }\fi
\providecommand{\bibinfo}[2]{#2}
\providecommand{\eprint}[2][]{\url{#2}}

\bibitem[{\citenamefont{Kearns et~al.}(1994)\citenamefont{Kearns, Mansour, Ron,
  Rubinfeld, Schapire, and Sellie}}]{kearns_learnability_1994}
\bibinfo{author}{\bibfnamefont{M.}~\bibnamefont{Kearns}},
  \bibinfo{author}{\bibfnamefont{Y.}~\bibnamefont{Mansour}},
  \bibinfo{author}{\bibfnamefont{D.}~\bibnamefont{Ron}},
  \bibinfo{author}{\bibfnamefont{R.}~\bibnamefont{Rubinfeld}},
  \bibinfo{author}{\bibfnamefont{R.~E.} \bibnamefont{Schapire}},
  \bibnamefont{and} \bibinfo{author}{\bibfnamefont{L.}~\bibnamefont{Sellie}},
  in \emph{\bibinfo{booktitle}{Proceedings of the twenty-sixth annual {ACM}
  symposium on {Theory} of {Computing}}} (\bibinfo{publisher}{Association for
  Computing Machinery}, \bibinfo{address}{New York, NY, USA},
  \bibinfo{year}{1994}), {STOC} '94, pp. \bibinfo{pages}{273--282}, ISBN
  \bibinfo{isbn}{978-0-89791-663-9}.

\bibitem[{\citenamefont{Biamonte et~al.}(2017)\citenamefont{Biamonte, Wittek,
  Pancotti, Rebentrost, Wiebe, and Lloyd}}]{biamonte_qml_2017}
\bibinfo{author}{\bibfnamefont{J.}~\bibnamefont{Biamonte}},
  \bibinfo{author}{\bibfnamefont{P.}~\bibnamefont{Wittek}},
  \bibinfo{author}{\bibfnamefont{N.}~\bibnamefont{Pancotti}},
  \bibinfo{author}{\bibfnamefont{P.}~\bibnamefont{Rebentrost}},
  \bibinfo{author}{\bibfnamefont{N.}~\bibnamefont{Wiebe}}, \bibnamefont{and}
  \bibinfo{author}{\bibfnamefont{S.}~\bibnamefont{Lloyd}},
  \bibinfo{journal}{Nature} \textbf{\bibinfo{volume}{549}},
  \bibinfo{pages}{195} (\bibinfo{year}{2017}).

\bibitem[{\citenamefont{Arunachalam and
  de~Wolf}(2017)}]{dewolf_qlearntheo_2017}
\bibinfo{author}{\bibfnamefont{S.}~\bibnamefont{Arunachalam}} \bibnamefont{and}
  \bibinfo{author}{\bibfnamefont{R.}~\bibnamefont{de~Wolf}},
  \bibinfo{journal}{arXiv:1701.06806}  (\bibinfo{year}{2017}).

\bibitem[{\citenamefont{Lloyd et~al.}(2013)\citenamefont{Lloyd, Mohseni, and
  Rebentrost}}]{lloyd_qml_2013}
\bibinfo{author}{\bibfnamefont{S.}~\bibnamefont{Lloyd}},
  \bibinfo{author}{\bibfnamefont{M.}~\bibnamefont{Mohseni}}, \bibnamefont{and}
  \bibinfo{author}{\bibfnamefont{P.}~\bibnamefont{Rebentrost}}
  (\bibinfo{year}{2013}), \bibinfo{note}{arXiv:1307.0411}.

\bibitem[{\citenamefont{Carleo et~al.}(2019)\citenamefont{Carleo, Cirac,
  Cranmer, Daudet, Schuld, Tishby, Vogt-Maranto, and
  Zdeborov\'a}}]{carleo_mlphys_2019}
\bibinfo{author}{\bibfnamefont{G.}~\bibnamefont{Carleo}},
  \bibinfo{author}{\bibfnamefont{J.~I.} \bibnamefont{Cirac}},
  \bibinfo{author}{\bibfnamefont{K.}~\bibnamefont{Cranmer}},
  \bibinfo{author}{\bibfnamefont{L.}~\bibnamefont{Daudet}},
  \bibinfo{author}{\bibfnamefont{M.}~\bibnamefont{Schuld}},
  \bibinfo{author}{\bibfnamefont{N.}~\bibnamefont{Tishby}},
  \bibinfo{author}{\bibfnamefont{L.}~\bibnamefont{Vogt-Maranto}},
  \bibnamefont{and}
  \bibinfo{author}{\bibfnamefont{L.}~\bibnamefont{Zdeborov\'a}},
  \bibinfo{journal}{Rev. Mod. Phys.} \textbf{\bibinfo{volume}{91}},
  \bibinfo{pages}{045002} (\bibinfo{year}{2019}).

\bibitem[{\citenamefont{Benedetti et~al.}(2019)\citenamefont{Benedetti, Lloyd,
  Sack, and Fiorentini}}]{benedetti_pqc_2019}
\bibinfo{author}{\bibfnamefont{M.}~\bibnamefont{Benedetti}},
  \bibinfo{author}{\bibfnamefont{E.}~\bibnamefont{Lloyd}},
  \bibinfo{author}{\bibfnamefont{S.}~\bibnamefont{Sack}}, \bibnamefont{and}
  \bibinfo{author}{\bibfnamefont{M.}~\bibnamefont{Fiorentini}},
  \bibinfo{journal}{Quantum Sc. Tech.} \textbf{\bibinfo{volume}{4}},
  \bibinfo{pages}{043001} (\bibinfo{year}{2019}).

\bibitem[{\citenamefont{Cerezo et~al.}(2021)\citenamefont{Cerezo, Arrasmith,
  Babbush, Benjamin, Endo, Fujii, McClean, Mitarai, Yuan, Cincio
  et~al.}}]{cerezo2021variational}
\bibinfo{author}{\bibfnamefont{M.}~\bibnamefont{Cerezo}},
  \bibinfo{author}{\bibfnamefont{A.}~\bibnamefont{Arrasmith}},
  \bibinfo{author}{\bibfnamefont{R.}~\bibnamefont{Babbush}},
  \bibinfo{author}{\bibfnamefont{S.~C.} \bibnamefont{Benjamin}},
  \bibinfo{author}{\bibfnamefont{S.}~\bibnamefont{Endo}},
  \bibinfo{author}{\bibfnamefont{K.}~\bibnamefont{Fujii}},
  \bibinfo{author}{\bibfnamefont{J.~R.} \bibnamefont{McClean}},
  \bibinfo{author}{\bibfnamefont{K.}~\bibnamefont{Mitarai}},
  \bibinfo{author}{\bibfnamefont{X.}~\bibnamefont{Yuan}},
  \bibinfo{author}{\bibfnamefont{L.}~\bibnamefont{Cincio}},
  \bibnamefont{et~al.}, \bibinfo{journal}{Nature Rev. Phys.}
  \textbf{\bibinfo{volume}{3}}, \bibinfo{pages}{625} (\bibinfo{year}{2021}).

\bibitem[{\citenamefont{Servedio and
  Gortler}(2004)}]{servedio_equivalences_2004}
\bibinfo{author}{\bibfnamefont{R.~A.} \bibnamefont{Servedio}} \bibnamefont{and}
  \bibinfo{author}{\bibfnamefont{S.~J.} \bibnamefont{Gortler}},
  \bibinfo{journal}{SIAM J. Comp.} \textbf{\bibinfo{volume}{33}},
  \bibinfo{pages}{1067} (\bibinfo{year}{2004}).

\bibitem[{\citenamefont{Dunjko et~al.}(2018)\citenamefont{Dunjko, Liu, Wu, and
  Taylor}}]{dunjko_exponential_2018}
\bibinfo{author}{\bibfnamefont{V.}~\bibnamefont{Dunjko}},
  \bibinfo{author}{\bibfnamefont{Y.-K.} \bibnamefont{Liu}},
  \bibinfo{author}{\bibfnamefont{X.}~\bibnamefont{Wu}}, \bibnamefont{and}
  \bibinfo{author}{\bibfnamefont{J.~M.} \bibnamefont{Taylor}}
  (\bibinfo{year}{2018}), \bibinfo{note}{arXiv:1710.11160}.

\bibitem[{\citenamefont{Liu et~al.}(2021)\citenamefont{Liu, Arunachalam, and
  Temme}}]{liu_rigorous_2021}
\bibinfo{author}{\bibfnamefont{Y.}~\bibnamefont{Liu}},
  \bibinfo{author}{\bibfnamefont{S.}~\bibnamefont{Arunachalam}},
  \bibnamefont{and} \bibinfo{author}{\bibfnamefont{K.}~\bibnamefont{Temme}},
  \bibinfo{journal}{Nature Phys.} \textbf{\bibinfo{volume}{17}},
  \bibinfo{pages}{1013} (\bibinfo{year}{2021}).

\bibitem[{\citenamefont{Sweke et~al.}(2021)\citenamefont{Sweke, Seifert,
  Hangleiter, and Eisert}}]{sweke_generator_2021}
\bibinfo{author}{\bibfnamefont{R.}~\bibnamefont{Sweke}},
  \bibinfo{author}{\bibfnamefont{J.-P.} \bibnamefont{Seifert}},
  \bibinfo{author}{\bibfnamefont{D.}~\bibnamefont{Hangleiter}},
  \bibnamefont{and} \bibinfo{author}{\bibfnamefont{J.}~\bibnamefont{Eisert}},
  \bibinfo{journal}{Quantum} \textbf{\bibinfo{volume}{5}}, \bibinfo{pages}{417}
  (\bibinfo{year}{2021}).

\bibitem[{\citenamefont{Gyurik and Dunjko}(2022)}]{onseparations}
\bibinfo{author}{\bibfnamefont{C.}~\bibnamefont{Gyurik}} \bibnamefont{and}
  \bibinfo{author}{\bibfnamefont{V.}~\bibnamefont{Dunjko}},
  \emph{\bibinfo{title}{On establishing learning separations between classical
  and quantum machine learning with classical data}} (\bibinfo{year}{2022}),
  \bibinfo{note}{arXiv:2208.06339}.

\bibitem[{\citenamefont{Valiant}(1984)}]{valiant_pac_1984}
\bibinfo{author}{\bibfnamefont{L.~G.} \bibnamefont{Valiant}},
  \bibinfo{journal}{Communications of the ACM} \textbf{\bibinfo{volume}{27}},
  \bibinfo{pages}{1134} (\bibinfo{year}{1984}).

\bibitem[{\citenamefont{Kearns and Vazirani}(1994)}]{kearns1994introduction}
\bibinfo{author}{\bibfnamefont{M.~J.} \bibnamefont{Kearns}} \bibnamefont{and}
  \bibinfo{author}{\bibfnamefont{U.}~\bibnamefont{Vazirani}},
  \emph{\bibinfo{title}{An introduction to computational learning theory}}
  (\bibinfo{publisher}{MIT press}, \bibinfo{year}{1994}).

\bibitem[{\citenamefont{Kearns}(1998)}]{sqmodel}
\bibinfo{author}{\bibfnamefont{M.}~\bibnamefont{Kearns}}, \bibinfo{journal}{J.
  ACM} \textbf{\bibinfo{volume}{45}}, \bibinfo{pages}{983–1006}
  (\bibinfo{year}{1998}).

\bibitem[{\citenamefont{Xiao}(2010)}]{xiao_learning_2010}
\bibinfo{author}{\bibfnamefont{D.}~\bibnamefont{Xiao}}, in
  \emph{\bibinfo{booktitle}{{COLT} 2010 - {The} 23rd {Conference} on {Learning}
  {Theory}, {Haifa}, {Israel}, {June} 27-29, 2010}}, edited by
  \bibinfo{editor}{\bibfnamefont{A.~T.} \bibnamefont{Kalai}} \bibnamefont{and}
  \bibinfo{editor}{\bibfnamefont{M.}~\bibnamefont{Mohri}}
  (\bibinfo{publisher}{Omnipress}, \bibinfo{year}{2010}), pp.
  \bibinfo{pages}{516--528}.

\bibitem[{\citenamefont{Goldreich et~al.}(1986)\citenamefont{Goldreich,
  Goldwasser, and Micali}}]{goldreich_randomfunctions_1986}
\bibinfo{author}{\bibfnamefont{O.}~\bibnamefont{Goldreich}},
  \bibinfo{author}{\bibfnamefont{S.}~\bibnamefont{Goldwasser}},
  \bibnamefont{and} \bibinfo{author}{\bibfnamefont{S.}~\bibnamefont{Micali}},
  \bibinfo{journal}{Journal of the ACM} \textbf{\bibinfo{volume}{33}},
  \bibinfo{pages}{792} (\bibinfo{year}{1986}).

\bibitem[{\citenamefont{Bogdanov and
  Rosen}(2017)}]{bogdanov_pseudo-random_2017}
\bibinfo{author}{\bibfnamefont{A.}~\bibnamefont{Bogdanov}} \bibnamefont{and}
  \bibinfo{author}{\bibfnamefont{A.}~\bibnamefont{Rosen}},
  \emph{\bibinfo{title}{Pseudorandom functions: Three decades later}}
  (\bibinfo{publisher}{Springer}, \bibinfo{address}{Berlin},
  \bibinfo{year}{2017}).

\bibitem[{\citenamefont{Cross et~al.}(2015)\citenamefont{Cross, Smith, and
  Smolin}}]{cross_quantumlpn_2015}
\bibinfo{author}{\bibfnamefont{A.~W.} \bibnamefont{Cross}},
  \bibinfo{author}{\bibfnamefont{G.}~\bibnamefont{Smith}}, \bibnamefont{and}
  \bibinfo{author}{\bibfnamefont{J.~A.} \bibnamefont{Smolin}},
  \bibinfo{journal}{Phys. Rev. A} \textbf{\bibinfo{volume}{92}},
  \bibinfo{pages}{012327} (\bibinfo{year}{2015}).

\end{thebibliography}

\end{document}